\documentclass[journal]{IEEEtran}
\usepackage{amssymb,stmaryrd,amsmath,amsfonts,rotating}
\usepackage[noadjust]{cite} \usepackage{color} 
\usepackage[vflt]{floatflt} \usepackage{epic}
\usepackage{color}
\newcommand{\bb}{\underline{b}}

\newcommand{\x}{\underline{x}}
\newcommand{\llr}{\underline{l}}
\newcommand{\y}{\underline{y}}

\RequirePackage{bbm} 
\definecolor{TODO}{rgb}{0.6,0.6,0.6} 

\definecolor{TOCHECK}{rgb}{0.8,0.8,0.8} 

\newtheorem{theorem}{Theorem}
\newcommand{\btheo}{\begin{theorem}}
\newcommand{\etheo}{\end{theorem}}
\newcommand{\bproof}{\begin{proof}}
\newcommand{\eproof}{\end{proof}}
\newtheorem{definition}[theorem]{Definition}
\newcommand{\bdefi}{\begin{definition}}
\newcommand{\edefi}{\end{definition}}
\newtheorem{fact}[theorem]{Fact}
\newcommand{\bprop}{\begin{fact}}
\newcommand{\eprop}{\end{fact}}
\newtheorem{corollary}[theorem]{Corollary}
\newcommand{\bcor}{\begin{corollary}}
\newcommand{\ecor}{\end{corollary}}
\newtheorem{example}[theorem]{Example}
\newcommand{\bex}{\begin{example}}
\newcommand{\eex}{\end{example}}
\newtheorem{lemma}[theorem]{Lemma}
\newcommand{\blemma}{\begin{lemma}}
\newcommand{\elemma}{\end{lemma}}
\newtheorem{remark}[theorem]{Remark}
\newcommand{\bremark}{\begin{remark}}
\newcommand{\eremark}{\end{remark}}
\newtheorem{conj}[theorem]{Conjecture}
\newcommand{\bconj}{\begin{conj}}
\newcommand{\econj}{\end{conj}}

\def\0{{\tt 0}} 
\def\1{{\tt 1}} 
\def\?{{\tt *}} 
\newcommand{\qed}{{\hfill \footnotesize $\blacksquare$}}
\renewcommand{\mid}{\,|\,}

 \allowdisplaybreaks
\begin{document} 
\title{Improved Linear Programming Decoding using Frustrated Cycles} 
\author{\authorblockN{Shrinivas Kudekar\authorrefmark{1},  Jason K. Johnson\authorrefmark{1} and Misha Chertkov\authorrefmark{1}\authorrefmark{2}\\ } \authorblockA{\authorrefmark{1}
Center for Nonlinear Studies \\ \& Theoretical Division T-4 \\
Los Alamos National Laboratory, Los Alamos NM, USA. \\
Email: \{skudekar, jasonj, chertkov\}@lanl.gov}\\
\authorblockA{\authorrefmark{2} New Mexico Consortium, Los Alamos, NM, USA.}
}

\maketitle \begin{abstract} We consider transmission over a binary-input
additive white Gaussian noise channel using low-density parity-check codes. One
of the most popular techniques for decoding low-density parity-check codes is
the linear programming decoder. In general, the linear programming decoder is
suboptimal.  I.e., the word error rate is higher than the optimal, maximum a
posteriori decoder.

In this paper we present a systematic approach to enhance the linear program
decoder. More precisely, in the cases where the linear program outputs a
fractional solution, we give a simple algorithm to  identify {\em frustrated
cycles} which cause the output of the linear program to be fractional. Then
adding these cycles, {\em adaptively} to the basic linear program, we show {\em
improved word error rate performance}. 

\end{abstract}

\section{Introduction} We consider transmission over a binary-input additive white Gaussian  
 noise channel (BIAWGNC) using low-density parity-check (LDPC) codes. The
two most fundamental decoders in this context are the belief propagation (BP)
decoder \cite{RiU08} and the linear programming (LP) decoder \cite{FWK05}.  In
this paper we are interested in the performance of the LP decoder.  There is an
extensive literature on analysis and design of the LP decoder for LDPC codes
 \cite{FWK05, Fel03, KoV03b, DDKW08}. As is well known, LP decoders have the advantage
 that they provide the {\em ML certificate}. This means that, if the LP decoder
outputs an integer solution, then it must be the maximum likelihood (ML)
codeword. Thus in this case the LP behaves as an optimal decoder. One can also
say that in this case there is no {\em duality gap}. 

However, it is also known that in general the LP decoder is suboptimal
\cite{FWK05}. I.e., there exists channel noise realizations such that the LP
decoder outputs a fractional solution, known as {\em pseudocodewords}
\cite{KoV03b}, but still there exists a unique codeword which
{\em minimizes} the objective function. This implies that the LP decoder is not
successful in finding the ML codeword. As a result, there is a gap between the
performance of the LP decoder and the ML decoder. Hence it is an interesting  
 question to understand what causes the LP decoder to fail and further
if there exists methods to {\em improve} the LP decoder. It is well known that
adding redundant parity-check nodes to the Tanner graph of the LDPC code
improves the LP decoder \cite{FWK05, Burshtein10}.  However it is not desirable to add all such
constraints as it will slow down the LP decoder considerably.  

In this work we propose an approach to adaptively add constraints to the LP
decoder which, simultaneously, reduce the duality gap and are tractable (i.e.,
the number of such additional constraints are small and also each constraint
involves only a small number of variables).  Such approaches, which try to get
rid off the fractional solution (or make the LP polytope tighter), have been
used to improve the LP decoding of LDPC codes \cite{FWK05, TagSie08, DimWain09,
Burshtein10, DYW07, Chertkov07, YFW06}. The new LP decoder which we propose,
identifies {\em frustrated cycles} (see Section~\ref{sec:LPwithFC}) when the
basic LP produces a fractional solution.  We show that these frustrated cycles
are the cause of inconsistency in the solution. Then we adaptively add them, as
constraints, to the basic LP decoder.  This enables us to recover the
 transmitted codeword in many cases. We show empirically that the new LP
decoder has an improved word error rate performance. Furthermore, the new LP
decoder also has tractable complexity.

\section{Channel Model, Maximum Likelihood Decoder and Linear Programming decoder} 

\subsection{Setup and Nomenclature} We consider transmission over a BIAWGNC
with  noise distribution given by $\mathcal{N}(0,\sigma^2)$. We use blocklength
$n$ LDPC encoding and denote $\x=\{x_1,x_2,\dots,x_n\}$ as the transmitted
codeword. The input codebit takes value in $\{0,1\}$. The received message is
denoted by $\y\in \mathbb{R}^n$.  We will use the loglikelihood ratio (LLR) to
represent the channel observations. More precisely, we have $l_i =
\log\frac{p_{y\vert x}(y_i \mid 0)}{p_{y\vert x}(y_i \mid 1)}$, where
$p_{y\vert x}(y\vert x)$ is the channel transition pdf. Let $\llr$ represent
the vector of LLRs.  

The LDPC code is represented by the usual Tanner graph representation\cite{RiU08}. Throughout
the paper we will use $(d_l, d_r)$-regular LDPC code ensembles to demonstrate our
 approach.  The design rate of the LDPC code is given by $1 -d_l/d_r$. In the
experiments we perform later, we consider the random $(3, 4)$-regular LDPC code
ensemble and the fixed $155$-Tanner code \cite{TSF01} which has degree 3
variable nodes and degree 5 check nodes.  We use $V$ to denote the set of $n$
variable nodes or codebits and $C$ to denote the set of $m$ parity check nodes.
A generic variable node and a check node is denoted by the letter $i$ and $c$
respectively.  Let $\mathcal{C}$ represent the code (or the set of codewords).  

\subsection{ML Decoder} The ML decoder
 can be written as the following combinatorial optimization problem \cite{FWK05},
 $\min_{\x\in \mathcal{C}} \sum_{i=1}^n l_i x_i.$
This is also the Integer  Program (IP) representing the ML decoding.  
 
\subsection{Basic Linear Programming Decoder}
For every check node $c\in C$, let $x_c=\{x_i \mid i\in c\}$. We also use $c\setminus i$ to denote the
 set of all variable nodes contained in check node $c$ except for the variable node $i$.
The above IP can be relaxed to 
\begin{align}
&\min_{\bb}  \sum_{i=1}^n \sum_{x_i\in\{0,1\}} l_i x_i b_i(x_i) \nonumber \\
& \text{s.t.} \quad \forall i\in V:\quad \sum_{x_i\in\{0,1\}}b_{i}(x_i)=1, \label{cond1} \nonumber \\
& \forall c\in C,\, \forall i\in c,\,  x_i\in \{0,1\}:\quad
 b_i(x_i)=\sum_{x_{c\setminus i}}b_{c}(x_i, x_{c\setminus i}) \nonumber \\
 & \forall c\in C,\, \forall x_c\,\,\text{s.t.}\,\,\sum_{i\in c}x_i=1,\,\,
 b_{c}(x_c)=0,\,\, \text{(local codeword)}\nonumber \\ 
& 0\leq b_i(x_i) \leq 1, \,\, \forall i\in V, \quad 0\leq b_{c}(x_c) \leq 1, \,\, \forall c\in C, 
\nonumber
\end{align}
which constitutes the standard LP decoder \cite{FWK05}.  
Here $b_i(x_i)$ represents the ``belief'' of the variable node  $i$ and
$b_c(x_c)$ represents the ``belief'' associated to the check node $c$.  
In the sequel, we will also say that $b_i(x_i)$ is the belief associated to the 
 singleton clique $i$ and $b_c(x_c)$ is the belief associated to a higher order 
 {\em clique}\footnote{In a {\em clique}, every node is connected to every other node. The LPs given in this paper always have beliefs associated to cliques.}. 
Also, $\bb$ represents the 
 vector of all the variable node and check node beliefs. 
Note that the objective function
 represents the ``cost'' of decoding a bit to 0. This cost is reduced if the corresponding LLR is negative.
The second condition imposed by the LP above is the {\em consistency} condition. In the third condition, 
 the sum is over GF(2). 

\section{Main Results: Improved LP decoding}

As mentioned earlier, our approach is to adaptively add constraints to the LP
which decrease the duality gap. Furthermore, we want the number of such
additional constraints to be small and also each constraint to involve only a small
number of variables.


There are many existing approaches to improve the LP decoder 
 \cite{FWK05, TagSie08, DimWain09, Burshtein10, DYW07, Chertkov07,
YFW06}.  In \cite{FWK05} an improved LP decoder based on ``lift-and-project''
method was introduced.  In \cite{DimWain09}, the LP is enhanced by eliminating
the facet containing the fractional solution. In \cite{TagSie08, Burshtein10},
extra constraints are added by combining parity checks which correspond to
violated constraints to improve the LP performance.  In \cite{DYW07} a
mixed-integer LP was introduced by fixing the most ``uncertain'' bit of the
pseudocodeword.  In \cite{Chertkov07} an adaptive LP decoder was introduced
 based on {\em loop calculus}. Critical loops were identified and then
broken by fixing bits on the loop. In \cite{YFW06} a non-linear programming decoder
was designed for decoding LDPC codes.  
 
\subsection{LP Decoders using Frustrated Subgraphs}\label{sec:LPwithFC} Although
our approach is in the same spirit as aforementioned works, the main ideas are
very different and have their origins in \cite{JJthesis08} and \cite{JMW07}.
Similar ideas have been independently used in \cite{SonJaa07, KoPaTz}. Before we
describe the basic idea let us first define the notion of a {\em frustrated
graph}.

\begin{definition}[Frustrated Graph] Consider a constraint
satisfaction problem (CSP) defined on $n$ binary (boolean) variables, $\x$, and
$m$ constraint nodes (each of which constraints a small set of variables). For each constraint $c$ there are
only certain configurations of $x_{c}\in \{0,1\}^{\vert c\vert}$ which satisfy it.  
Then, we say that the graph is {\em frustrated} if and only if there
is no assignment of $\x$ which satisfies all $m$ constraint nodes
simultaneously. \qed 
\end{definition}

Let us now define a CSP for our set-up. 

\begin{definition}[CSP obtained by the LP Solution] Assume that the output
of LP, \bb, is a fractional solution, i.e., we have a duality gap.  
 For every clique $c$ (with size at least two), the set of $x_c$ which {\em satisfy} the clique are those for which
 $b_c(x_c)>0$. In other words, the set of $x_c$ satisfying the clique $c$, correspond to the
 support set of $b_c(x_c)$. Consequently, the CSP is given by the $n$ (binary) variables, $\{x_i\}_{i=1}^n$
 and the set of cliques $c$ (constraining the variables as described previously).
  \qed 
\end{definition}

We now show that if the output of the LP has a frustrated subgraph, then it
must have a duality gap, i.e., the solution must be fractional. 

\begin{lemma}
If there exists a frustrated subgraph, then there
 is a duality gap. 
\end{lemma}
\begin{proof}
Indeed, suppose on the contrary there was no duality gap, i.e, output of the LP
is integral. Thus for every clique $c$ (singleton or higher order),
$b_c(x_c)=1$ for some $x_c\in \{0,1\}^{\vert c\vert}$ and $b_c(x_c)=0$ for the rest.
Consider any subset of the cliques, $\mathcal{C}=\{c_1, c_2,\dots, c_r\}$.
Let $x^*_{c_i}$ be such that $b_{c_i}(x^*_{c_i})=1.0$. We
claim that $\cup_{i=1}^r x^*_{c_i}$ satisfies the CSP represented by
$\mathcal{C}$.  Indeed, this follows from the consistency imposed by the LP. 
Thus no subgraph is frustrated.
\end{proof}

Thus our strategy is as follows: first identify a frustrated subgraph from the
output of the basic LP; if we add this frustrated subgraph as a constraint in
our LP, then we ensure that this subgraph cannot be frustrated.  In our
experiments we see that, in many cases, adding the frustrated subgraphs
eliminates the duality gap. 

To ensure that the subgraph we add as a constraint to the LP becomes consistent
(or is not frustrated), we need to add all its maximal cliques and their
intersections to the LP. More precisely, we add the maximal cliques of the
junction tree\footnote{ See \cite{JJthesis08} for a discussion on Junction
trees. It can be shown that running LP on the junction tree of a graph is
optimal (equal to the original combinatorial optimization problem).  If frustrated
 subgraph is a cycle then we just add all the triangles which chordalizes the cycle.}
of that subgraph as extra beliefs to the LP. 

The main challenge that remains is to find a frustrated subgraph in tractable
time. In general, it is hard to find an arbitrary subgraph which is frustrated.
We also remark that in \cite{JJthesis08} it was found empirically that the
random field ising model could typically be solved (duality gap eliminated) by
adding frustrated cycles arising in the LP solution.  It is also known from
Barahona's work (see references within \cite{SonJaa07}) that adding cycles is
sufficient to solve the zero-field planar ising model.  Hence as a first step,
we focus on finding {\em frustrated cycles} of the graph. Frustrated cycles and
a procedure to find them are described in the next section. The procedure is
tractable and uses the implication graph method (to solve 2SAT problem) of
\cite{APT79, JJthesis08}.  For details see
Appendix B in \cite{JJthesis08}. 

\subsection{Implication Graph and Frustrated Cycles}\label{sec:implicationgraph}
For every clique $c$, consider all the two-projections of its belief. I.e., for every
$b_{c}(x_c)$, consider all the $b_{ij}(x_i, x_j) \,\,\forall\,\,i,j \in c$.
These are obtained by summing out the other variables. We construct the
implication graph as follows.  In the implication graph each node $i$ is present
as $i_+$ (for $x_i=0$) and $i_-$ (for $x_i=1$). Thus, the implication graph has
a total of $2n$ nodes.  There is a directed edge present between $i$ and $j$
which represents the logical implication obtained from $b_{ij}(x_i,x_j)$. Let us
explain this in more details.  To generate the logical implication, consider the
set $T$ of configurations of $(x_i, x_j)$ which render $b_{ij}(x_i,x_j)>0$ and
can introduce inconsistency.  Thus, $T$ is any of the following $(01,10),
(01,10,11), (01,10,00), (00,11), (00,11,10)$ and $(00,11,01)$. Indeed, moments
thought shows that other configurations, e.g., $(00,01,10,11)$,  are not
restrictive and hence do not  form any logical implication. Also, nodes which
have integer beliefs are present as isolated nodes in the graph and do not have
any edges entering or leaving it. 
Draw the directed edges using this $T$. E.g.,  suppose that LP outputs
beliefs such that $b_{ij}(0,1)>0, b_{ij}(1,0)>0, b_{ij}(1,1)>0, b_{i,j}(0,0)=0$
then $T=(01,10,11)$. This implies a directed edge from $i_+\to j_-$ and $j_+
\to i_-$, because if $x_i=0$ then we must have $x_j=1$ and if $x_j=0$ then
$x_i=1$.  
In figure~\ref{fig:implications}
 we illustrate all possible implications which form the building blocks for
 constructing the implication graph.
\begin{figure}[hbt] \centering
\setlength{\unitlength}{1.2bp}%
\begin{picture}(160,140)
\put(60,0){\includegraphics[scale=1.8]{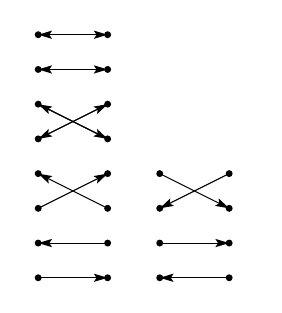}}
\put(10,130){\makebox(0,0){$x_i$}}
\put(22,130){\makebox(0,0){$x_j$}}
\put(10,118){\makebox(0,0){$0$}}
\put(20,118){\makebox(0,0){$0$}}
\put(10,107){\makebox(0,0){$1$}}
\put(20,107){\makebox(0,0){$1$}}
\put(70,120){\makebox(0,0){${\footnotesize i_+}$}}
\put(70,105){\makebox(0,0){${\footnotesize i_-}$}}
\put(115,120){\makebox(0,0){${\footnotesize j_+}$}}
\put(115,105){\makebox(0,0){${\footnotesize j_-}$}}
\put(70,90){\makebox(0,0){${\footnotesize i_+}$}}
\put(70,75){\makebox(0,0){${\footnotesize i_-}$}}
\put(115,90){\makebox(0,0){${\footnotesize j_+}$}}
\put(115,75){\makebox(0,0){${\footnotesize j_-}$}}
\put(70,60){\makebox(0,0){${\footnotesize i_+}$}}
\put(70,45){\makebox(0,0){${\footnotesize i_-}$}}
\put(115,60){\makebox(0,0){${\footnotesize j_+}$}}
\put(115,45){\makebox(0,0){${\footnotesize j_-}$}}
\put(70,30){\makebox(0,0){${\footnotesize i_+}$}}
\put(70,15){\makebox(0,0){${\footnotesize i_-}$}}
\put(115,30){\makebox(0,0){${\footnotesize j_+}$}}
\put(115,15){\makebox(0,0){${\footnotesize j_-}$}}

\put(10,89){\makebox(0,0){$0$}}
\put(20,89){\makebox(0,0){$1$}}
\put(10,77){\makebox(0,0){$1$}}
\put(20,77){\makebox(0,0){$0$}}

\put(5,60){\makebox(0,0){$0$}}
\put(15,60){\makebox(0,0){$0$}}
\put(5,52){\makebox(0,0){$1$}}
\put(15,52){\makebox(0,0){$0$}}
\put(5,44){\makebox(0,0){$0$}}
\put(15,44){\makebox(0,0){$1$}}
\put(30,60){\makebox(0,0){$1$}}
\put(40,60){\makebox(0,0){$1$}}
\put(30,52){\makebox(0,0){$1$}}
\put(40,52){\makebox(0,0){$0$}}
\put(30,44){\makebox(0,0){$0$}}
\put(40,44){\makebox(0,0){$1$}}
\put(27,52){\makebox(0,0){$\Bigg($ }}
\put(47,52){\makebox(0,0){$\Bigg)$ }}
\put(125,52){\makebox(0,0){$\bigg($ }}
\put(167,52){\makebox(0,0){$\bigg)$ }}

\put(5,30){\makebox(0,0){$0$}}
\put(15,30){\makebox(0,0){$0$}}
\put(5,22){\makebox(0,0){$1$}}
\put(15,22){\makebox(0,0){$1$}}
\put(5,14){\makebox(0,0){$0$}}
\put(15,14){\makebox(0,0){$1$}}
\put(30,30){\makebox(0,0){$0$}}
\put(40,30){\makebox(0,0){$0$}}
\put(30,22){\makebox(0,0){$1$}}
\put(40,22){\makebox(0,0){$1$}}
\put(30,14){\makebox(0,0){$1$}}
\put(40,14){\makebox(0,0){$0$}}
\put(27,22){\makebox(0,0){$\Bigg($ }}
\put(47,22){\makebox(0,0){$\Bigg)$ }}
\put(125,22){\makebox(0,0){$\bigg($ }}
\put(167,22){\makebox(0,0){$\bigg)$ }}

\end{picture}
\caption{Figure shows all possible implications between $x_i$ and $x_j$. These are used as basic 
 building blocks to create the implication graph. 
\label{fig:implications}}
\end{figure}
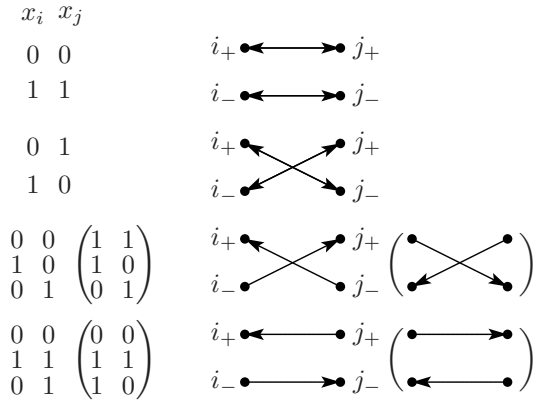 

Finally, a frustrated cycle is defined to be a {\em directed cycle or a directed path}
which visits both $i_+$ and $i_-$, once, for any $i$. One can find all such cycles
and paths in a time which is linear in the number of nodes of the implication
graph.

Figure~\ref{fig:frustratedcycles} shows the possible frustrated cycles which
are obtained from the implication graph.
 The figure on the left shows true
frustration. I.e., from the logical implications, obtained by the LP solution, we have that
$x_i=0$ implies $x_i=1$ and $x_i=1$ implies $x_i=0$. This means that the set of
local beliefs (which lie on the cycle connecting $i_+$ to $i_-$) are not
consistent. Hence it naturally suggests that there is frustration in the LP
solution. The other kind of frustration, suggested by the remaining figures, is
called as {\em quasi-frustration}. The figure in the middle demonstrates that
$x_i=1$ implies that $x_i=0$ but not the other way around. This
quasi-frustration implies that there cannot be a global joint distribution (on
all the variable nodes) such that it is consistent with the local beliefs.
Indeed, if it were true, then we know that it must assign $b_i(x_i=0) > 0$ and
$b_i(x_i=1) > 0$. This is because the variable node $i$ is present in the
implication graph and hence must have a fractional solution for $b_i(x_i)$.
However, from the implication graph $x_i=1$ implies $x_i=0$, hence any
configuration (on all nodes), which has a non-zero probability, cannot have $x_i =
1$, i.e., $b_i(x_1=1)=0$, a contradiction. 
\begin{figure}[hbt] \centering
\setlength{\unitlength}{1.2bp}%
\begin{picture}(160,53)
\put(0,0){\includegraphics[scale=1.2]{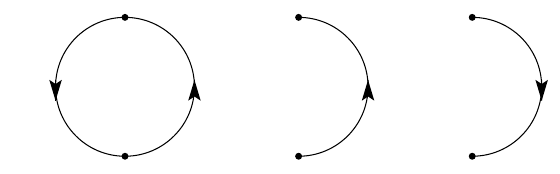}}
\put(35,50){\makebox(0,0){\footnotesize $i_+$}}
\put(35,2){\makebox(0,0){\footnotesize $i_-$}}
\put(85,50){\makebox(0,0){\footnotesize $i_+$}}
\put(85,2){\makebox(0,0){\footnotesize $i_-$}}
\put(135,50){\makebox(0,0){\footnotesize $i_+$}}
\put(135,2){\makebox(0,0){\footnotesize $i_-$}}
\end{picture}
\caption{ Figure shows the possible frustrated cycles present in the implication graph. 
 The first cycle is truly frustrated, since we must have $x_i=0$ implies $x_i=1$
 and vice-versa.  The remaining two cycles are quasi-frustrated, since either $x_i=0$ implies $x_i=1$ or
 vice-versa, but not both at the same time. These cycles are added to the LP and the enhanced decoder is termed LP-Frustrated Cycles (LP-FC). 
\label{fig:frustratedcycles}}
\end{figure}
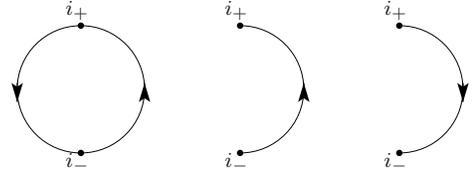 
We remark here that once we have found a frustrated cycle on the implication
graph, one can easily obtain the cycle on the original graph, by just projecting
the nodes on the implication graph back to the nodes on the original graph. 
The method in which we add the frustrated cycle to the LP is illustrated
 in the example below. 

\begin{example}[Triangulation of Frustrated Cycles]\label{ex:triangulation}
Figure~\ref{fig:triangulation} shows a cycle $(x_1, x_2, x_3, x_4, x_5, x_6,
x_7, x_8)$ which we add to the LP as a constraint. Adding the entire belief,
$b(x_1, x_2,\dots,x_8)$, as a constraint,  would be expensive and result in
$2^8$ extra variables and constraints amongst them.  Instead we add the maximal
cliques of its junction tree.  To do this, we first chordalize or triangulate
the cycle, as shown in the figure~\ref{fig:triangulation}, into the 6 triangles given by $(x_1, x_2,
x_3), (x_1, x_3, x_4), (x_1, x_4, x_5), (x_1, x_5, x_6)$, $(x_1, x_6, x_7)$,
$(x_1, x_7, x_8)$. These triangles are the maximal cliques and we add them as
constraints to the LP. E.g., we add $b_{x_1x_2x_3}(x_1,x_2,x_3)$ for all
$x_1,x_2,x_3\in \{0,1\}$. For every belief that we add to the LP, we add
constraints to ensure consistency with previously added beliefs. E.g., when we
add $b_{x_1,x_2,x_3}(x_1,x_2,x_3)$ and $b_{x_1,x_3,x_4}(x_1,x_3,x_4)$ we
introduce the constraint $\sum_{x_2} b_{x_1,x_2,x_3}(x_1, x_3) =
\sum_{x_4}b_{x_1,x_3,x_4}(x_1,x_3,x_4)$ for all values of $x_1, x_3$. In other words, every clique that we
add to the LP, must be consistent across its intersections.  
\end{example} 

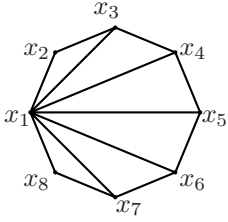
\begin{figure}[hbt] \centering
\setlength{\unitlength}{1.2bp}%
\begin{picture}(160,60)
\put(40,0){\includegraphics[scale=0.8]{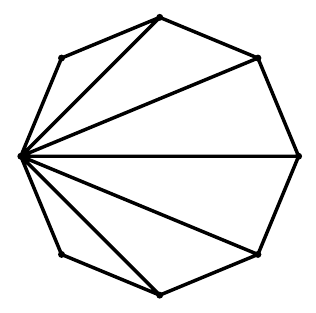}}
\put(40,28){\makebox(0,0){$x_1$}}
\put(46,49){\makebox(0,0){$x_2$}}
\put(68,62){\makebox(0,0){$x_3$}}
\put(95,49){\makebox(0,0){$x_4$}}
\put(102,28){\makebox(0,0){$x_5$}}
\put(95,8){\makebox(0,0){$x_6$}}
\put(75,0){\makebox(0,0){$x_7$}}
\put(46,8){\makebox(0,0){$x_8$}}
\end{picture}
\caption{Figure shows triangulation of the cycle $(x_1,x_2,x_3,x_4,x_5,x_6,x_7,x_8,x_1)$. The triangles chordalize the cycle and 
 form the maximal cliques. The details are explained in example~\ref{ex:triangulation}. 
\label{fig:triangulation}}
\end{figure}

\subsection{Experiments using Frustrated Cycles}\label{sec:LPwithcycles}

We consider BIAWGNC where the standard deviation of the noise is denoted by $\sigma$. We consider
two types of LDPC encoding: (i) regular $(3,4)$ LDPC ensemble with design rate
equal to $1/4$ and the (ii) 155-Tanner code \cite{TSF01}. The 155-Tanner code
 has a design rate of $2/5$. We let the standard deviation of the
noise, $\sigma$, take values in the set $\{0.5, 0.6, 0.7, 0.8, 0.9, 0.95, 1.00,
1.05, 1.10, 1.15, 1.20\}$. We run 2000 trials for each value of $\sigma$. We run
experiments for both the $(3,4)$-regular ensemble and the 155-Tanner code. For
the $(3,4)$-regular ensemble, in each trial a code is generated uniformly at random and used
for transmission.

Since we are transmitting over a symmetric channel, for the purpose of performance
analysis we can assume that we are transmitting the all-zero codeword
\cite{RiU08}. Under this assumption, the distribution of the LLRs are given by
$\mathcal{N}(\frac2{\sigma^2}, \frac4{\sigma^2})$. The generated LLRs are 
fed to both basic LP and LP-FC decoder. The LP-FC algorithm is described below.  For
any decoder, if the output equals the all-zero codeword, we declare success,
else there is an error. We plot the word error rate (WER) versus the SNR
($E_b/N_0$) in dB.

\begin{center}\framebox[0.95\columnwidth]{
\begin{minipage}{0.90\columnwidth}
\vspace{2mm}
\underline{LP-FC Decoder:}
\vspace{2mm}
\begin{enumerate}
\item  Run the basic LP. Go to step 4.
\item  If the output is fractional, find the frustrated cycle (FC) of the
smallest length and add all its triangles. 
\item  Rerun the LP.                                                    
\item If output is integral, stop else go to 2.                    
\end{enumerate}
\vspace{1mm}
\end{minipage}}
\end{center}
\vspace{3mm}

\subsubsection{Experiments with $(3,4)$-regular LDPC ensemble}

Figure~\ref{fig:3_4_160} shows the performance curve when we use the
$(3,4)$-regular ensemble with blocklength 160.  The dark curve represents the
performance (averaged over 2000 trials where in each trial a code and noise
realization is picked uniformly at random) when we use the basic LP decoder. The gray curve
denotes the performance under LP-FC. We remark here that for each simulation
trial, the LP and LP-FC were run on the same code and noise realization. We
observe that there were many trials where the basic LP decoder failed. However,
adding a small number of cycles to the LP helped in retrieving the transmitted
all-zero codeword. From the figures we observe that LP-FC performs much better
than the basic LP.


\begin{figure}[htp] \centering
\input{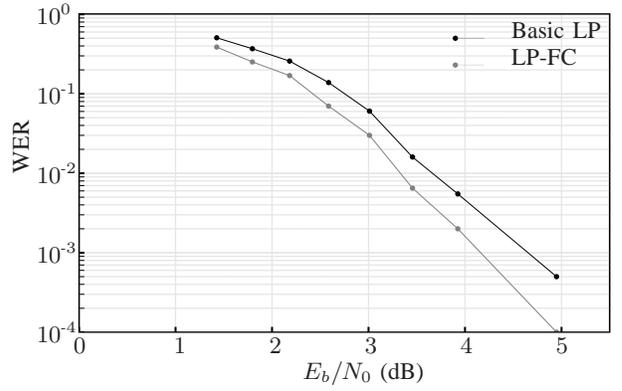} 
\caption{\label{fig:3_4_160} The figure shows the performance improvement of LP-FC over the basic LP. 
 In this experiment $(3,4)$-regular LDPC ensemble of blocklength 160 was used. The dark curve depicts 
 the word error rate (WER) performance of the basic LP and the gray curve shows the performance of the LP-FC. 
} 
\end{figure}

Table~\ref{tab:complexity} demonstrates various quantities for different values
of the SNR for the case when we use the $(3,4)$-regular LDPC ensemble with
blocklength 160. The second column shows the average number of LPs called in
the LP-FC algorithm, i.e, the number of times step 3 is called in the LP-FC
algorithm.  The remaining columns illustrate the complexity of the extra LPs
which are solved in the LP-FC algorithm.  The third and the sixth column show
the number of non-zeros in the constraint matrix and the dimensions of the
constraint matrix when the basic LP is run. The fourth and the last column show
the average number of non-zero entries in the constraint matrix and the average
dimensions of the of the constraint matrix, when the LP-FC algorithm is run,
respectively. Also shown in the fifth column is the maximum number of non-zero
entries in any constraint matrix which occurs in the LP-FC algorithm.  Thus, the
table demonstrates that the size of the LP, after adding the frustrated cycles,
does not increase by much. Hence the LP-FC decoder is kept tractable.

We also observe that every cycle we add is a {\em simple} cycle, without any
self-intersections. 
 
\begin{table} \centering 
\begin{tabular}{|p{0.6cm}|p{0.5cm}|p{0.7cm}|p{0.7cm}|p{0.8cm}|p{1.3cm}|p{1.3cm}|}
\hline
SNR (in dB) &  Num. of LPs (avg.) & Non-zeros (avg.) LP ($\times 10^4$) & Non-zeros (avg.) LP-FC ($\times 10^4$)& Non-zeros (max) LP-FC  ($\times 10^4$) & Dim. for LP (rows,cols) & Dim. (avg.) for LP-FC  (rows,cols) \\
\hline
 3.93 &  3 & 2.6760  & 2.7028  &  2.9252  &(6490,4920) &(6569,4970)  \\
 3.46 &  6 & 2.6597  & 2.7190  &  3.1780  &(6458,4896) &(6635,5009)  \\
 3.01 &  7 & 2.6694  & 2.7471  &  3.8020  &(6477,4910) &(6708,5059)  \\
 2.59 &  6 & 2.6613  & 2.7312  &  3.6116  &(6461,4899) &(6669,5032)  \\
 2.18 &  7 & 2.6637  & 2.7514  &  4.1012  &(6466,4902) &(6727,5070)  \\
 1.79 &  7 & 2.6659  & 2.7403  &  3.7796  &(6470,4905) &(6692,5047)  \\
 1.43 &  8 & 2.6572  & 2.7483  &  3.8984  &(6453,4893) &(6725,5068)   \\
\hline 
\end{tabular} 
\caption{ \label{tab:complexity} Complexity comparison of  LP and LP-FC decoders. }
\vspace{-0.5cm}
\end{table}

\subsubsection{Experiments with $155$-Tanner code \cite{TSF01}}
We also perform experiments with the 155-Tanner code which has 155 variable
nodes and 93 check nodes. The
experimental set-up is same as before. 

Figure~\ref{fig:155} shows the performance curve (averaged over 2000 noise
realizations for each value of $\sigma$) when we use the 155-Tanner code.  Again, we
observe that LP-FC performs much better than the basic LP. 

\begin{figure}[htp] \centering
\input{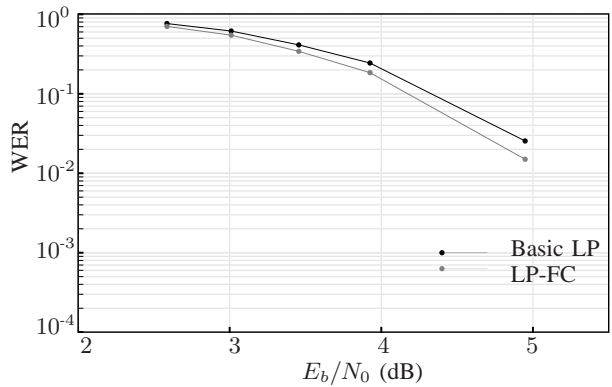} 
\caption{\label{fig:155} The figure shows the performance improvement of LP-FC over the basic LP
 when the 155-Tanner code was used.} 
\end{figure}

We also perform experiments at very high SNR for the 155-Tanner code. This
known as the error-floor regime. The error-floor occurs because of low-weight
pseudocodewords which are fractional, i.e., not codewords. In \cite{CS08} a
pseudocodeword search algorithm was used to generate  pseudocodewords which are
not codewords.  We pick 200 worst pseudocodewords which have effective weight
\cite{KoV03b} less than the minimum Hamming distance of 20. Also, all these
pseudocodewords will dominate the WER when SNR becomes very large.  

The experiment we perform is as follows. We take the corresponding noise
realizations which gave rise to these 200 pseudocodewords. We run the basic LP
on then and confirm that it fails on all these noise realizations and indeed we
recover the fractional pseudocodewords. On the same noise realizations, we also
 run the LP-FC.  Remarkably, the LP-FC is able to recover the correct
(all-zero) codeword for all the 200 worst-case noise realizations. Furthermore,
the step 3 in the LP-FC algorithm was just called once. The constraint matrix
for the basic LP is has 51,646 non-zeros entries and a dimension of (8618,
7006). On the other hand the enhanced LP has, on an average, 52,676 non-zeros
entries and an average dimension of (8925, 7163). Again, the LP-FC is kept
tractable.

\section{Discussion} In this work we present an improved LP decoder, called
LP-FC, based on frustrated cycles.  We show that the presence of frustration in
the output of the basic LP solution is the cause of inconsistency. We add these
frustrated cycles as constraints to the LP, thus enhancing it. We observe
empirically that the LP-FC decoder eliminates the duality gap, in a large
number of cases. Our simulations demonstrate that the LP-FC has a much better
performance compared to the basic LP introduced in \cite{FWK05}. 

This approach toward enhancing the basic LP decoder opens up many interesting
research directions. One direction is to investigate if one can add a {\em
frustrated subgraph}, which is not a cycle, to enhance the LP, when the
addition of cycles is not enough to eliminate the duality gap. The
reason we choose to add frustrated cycles, is that as mentioned in
Section~\ref{sec:implicationgraph}, the algorithm for finding such cycles is
simple. It is not clear if there exists simple algorithms to find minimal 
frustrated subgraphs. 


Recently, improved LP detectors based on frustrated cycles was also used in
\cite{KJC11} for 2DISI channel. One future research direction is to
investigate other combinatorial problems in graphical coding, e.g., minimum
pseudocodeword weight problem, minimum Hamming distance etc. 

Another future direction would be to develop distributed, i.e.,
message-passing, versions for the LP-FC. 

\section{Acknowledgments} Our work at LANL was carried out under the auspices
of the National Nuclear Security Administration of the U.S. Department of
Energy at Los Alamos National Laboratory under Contract No.  DE-AC52-06NA25396.
SK acknowledges support of NMC via the NSF collaborative grant CCF-0829945 on
``Harnessing Statistical Physics for Computing and Communications.'' 

\bibliographystyle{IEEEtran} 
\bibliography{lanl}
\end{document}